\documentclass[conference]{IEEEtran}
\usepackage[utf8]{inputenc}
\usepackage{graphicx}
\usepackage{siunitx}
\usepackage{bm}
\usepackage{amsfonts}
\usepackage{amssymb}
\usepackage{amsmath}
\usepackage{mathrsfs}
\usepackage{verbatim}
\usepackage{amsthm}

\usepackage{graphicx}
\usepackage{float}
\usepackage{booktabs}
\usepackage{tabularx}
\usepackage{lscape}
\usepackage{hyperref}

\newtheorem{prop}{Proposition}

\begin{document}
\title{Application of the Waveform Relaxation Technique to the Co-Simulation of Power Converter Controller and Electrical Circuit Models}

\author{
		\IEEEauthorblockN{
		M. Maciejewski\IEEEauthorrefmark{1}\IEEEauthorrefmark{2},
			I. Cortes Garcia\IEEEauthorrefmark{3}, 
			S. Schöps\IEEEauthorrefmark{3}, 
			B. Auchmann\IEEEauthorrefmark{2}\IEEEauthorrefmark{4},
			L. Bortot\IEEEauthorrefmark{2}, 
			M. Prioli\IEEEauthorrefmark{2}, 
			and A.P. Verweij\IEEEauthorrefmark{2}
		}
		\IEEEauthorblockA{\IEEEauthorrefmark{1}Łódź University of Technology, Łódź, Poland}
		\IEEEauthorblockA{\IEEEauthorrefmark{3}Technische Universität Darmstadt, Darmstadt, Germany}
		\IEEEauthorblockA{\IEEEauthorrefmark{4}Paul Scherrer Institut, Villigen, Switzerland}
		\IEEEauthorblockA{\IEEEauthorrefmark{2}CERN, Geneva, Switzerland, E-mail: michal.maciejewski@cern.ch}
	}
\maketitle

\begin{abstract}
In this paper we present the co-simulation of a PID class power converter controller and an electrical circuit by means of the waveform relaxation technique. The simulation of the controller model is characterized by a fixed-time stepping scheme reflecting its digital implementation, whereas a circuit simulation usually employs an adaptive time stepping scheme in order to account for a wide range of time constants within the circuit model. In order to maintain the characteristic of both models as well as to facilitate model replacement, we treat them separately by means of input/output relations and propose an application of a waveform relaxation algorithm. Furthermore, the maximum and minimum number of iterations of the proposed algorithm are mathematically analyzed. The concept of controller/circuit coupling is illustrated by an example of the co-simulation of a PI power converter controller and a model of the main dipole circuit of the Large Hadron Collider.
\end{abstract}

\IEEEpeerreviewmaketitle

\section{Introduction}
Superconducting magnets are used in modern high-energy particle accelerators to steer the trajectory of beams of charged particles. The current in the magnets has to follow the increase of the energy of particles. A dedicated closed feedback loop is implemented for the magnets to generate the desired magnetic field. On top of that, there are quench protection systems responsible for discharging the stored magnetic energy in the circuit in case of emergency. A quench is a local transition from a superconducting to a normal conducting state. After a quench, the magnetic energy is dissipated in a small fraction of the coil and without any counter-measures, such a fault would lead to irreversible damage of the magnet and consequently the circuit. Computer simulations play an important role in understanding these complex phenomena \cite{Bortot_2016ab}. Recently, a waveform relaxation algorithm has been implemented in the STEAM (Simulation of Transient Effects in Accelerator Magnets) project to perform field/circuit coupling for the magneto-thermal analysis of superconducting magnets and circuits \cite{Cortes-Garcia_2017ab}. The goal of the present work on waveform relaxation for controller-circuit coupling is to include controller models in the STEAM framework. This will allow us to study the circuit behavior on a wider range of potential failure scenarios.

In this paper we focus on the nominal operation of a circuit, coupling models of a power converter controller and circuit. We consider a PID type current controller of a power converter and a linear electrical circuit. Typically, the respective models of both, controller and circuit, are developed separately with different software packages. There are several motivations for that. Firstly, both models are solved with different time-stepping, i.e., the controller model is executed with fixed time step, whereas the circuit model employs an adaptive time-stepping algorithm in order to account for the transient behavior. Secondly, the controller is usually designed on the basis of the first-order equivalent electrical circuit models, similarly, a circuit driven by an ideal current source is satisfactory for a wide range of simulation scenarios. Nevertheless, a coupling between both models can provide more information on both controller and circuit behavior, especially for the relevant (nonlinear) failure scenarios, where ultimately partial differential equations are used to mathematically describe the quench.

In \cite{Panda_2015aa} authors employ waveform relaxation method to perform controller testing by means of simulations. This approach is an alternative to hardware in the loop simulations, however does not discuss the convergence properties. We propose the application of waveform relaxation \cite{Lelarasmee_1982ab,White_1985aa,Burrage_1995aa} to the coupled controller/circuit problem and study its features. Instead of sequentially executing controller and circuit simulator according to the sampling frequency, the problem is translated into an iterative process on larger time windows. Such an approach promises benefits in terms communication costs and thus eventually computational time. Ultimately the approach will allow to seamlessly include the action of a controller in a STEAM field-circuit simulation; see above.

The remainder of the paper is organized as follows. Section~II introduces the circuit and controller models. Section~III discusses the proposed waveform relaxation algorithm and studies convergence properties. In Section~IV we present a PI controller design procedure for the first-order model of the Large Hadron Collider (LHC) main dipole circuit. The paper is concluded with an application of the waveform relaxation scheme to the simulation of electrical transients in the superconducting circuits powered by a voltage source controlled by a PI controller. 
 \section{Mathematical Models}
In this section we present a general setting for treating controllers of PID class and linear circuits composed of resistors, inductors, and capacitors as ordinary differential equations (ODEs) and differential-algebraic equations (DAEs) on a time interval $\mathcal{I}=(0,T]$.

\subsection{PID Controller Model}
The governing equation of a continuous PID controller, with $k_{\text{p}}$, $k_{\text{i}}$, and $k_{\text{d}}$ denoting, respectively, proportional, integral, and differential gains, is given in continuous form as
\begin{equation}
U_{\text{con}}(t) = k_{\text{p}} E(t-t_\text{d}) + k_{\text{i}} \!\int_0^t\! E(\tau-t_\text{d}) d\tau + k_{\text{d}} \partial_t E(t-t_\text{d}),
\label{eq:ContinuousPIDController}
\end{equation}
where $t_\text{d}\geq0$ denotes a possible delay due to physical limitations of the controller, and $E(t)=Y_{\text{ref}}(t)-Y_{\text{meas}}(t)$ is the error given as a difference between the reference signal $Y_{\text{ref}}\in C^1(\mathcal{I})$ and its actual waveform $Y_{\text{meas}}\in C^1(\mathcal{I})$. The temporal discretization is given by times $t_j=jh_\text{con}$ with 
$j=1,\ldots,n_\text{con}$. The rectangle method for quadrature and the Implicit-Euler scheme for differentiation yield
\begin{equation}
\begin{aligned}
U_{\text{con}}(t_j) &= k_{\text{p}} E(t_j-t_\text{d}) + k_{\text{i}}h_{\text{con}} \sum_{i=1}^j E(t_i-t_\text{d}) +\\ &+k_{\text{d}} \frac 
{E(t_j-t_\text{d})-E(t_{j-1}-t_\text{d})}{h_{\text{con}}}.
\label{eq:DiscretizedPIDController}
\end{aligned}
\end{equation}
 
For the remainder of the paper let us consider the special case of a PI controller in order to analyze the convergence properties of the considered waveform relaxation method. Disregarding differential gains ($k_d=0$) and differentiating equation \eqref{eq:ContinuousPIDController} yields the following (delay) differential equation
\begin{equation}\label{eq:odecontroller}
\partial_t{U}_{\text{con}}(t) = k_{\text{p}} \partial_t{E}(t-t_\text{d}) + k_{\text{i}} E(t-t_\text{d})
\end{equation}
which is trivial in the sense that the right hand side does not depend on the unknown. It can be rewritten as
\begin{equation}
F_{\text{con}}(\partial_t{x}_{1},\partial_t{x}_2,{x}_2)=0,
\label{eq:ControllerInitialValueProblem}
\end{equation}
with the initial value $x_1(t_0)=x_{1,0}$, the unknown $x_1(t) = 
U_\text{con}(t)$, and the excitation $x_2(t) = Y_{\text{meas}}(t)$.
 
\subsection{Electrical Circuit Model}
Consider an electrical network composed of ${n}_{n}$ nodes with ${\phi}_{j}$ denoting the ${j}$-th nodal voltage with respect to ground, and 
${n}_{b}$ branches with ${I}_{i}$ and ${U}_{i}$ being the current through the ${i}$-th branch and the voltage across the ${i}$-th branch, respectively. 
The circuit topology is described by the 
incidence matrix $A$, with non-zero elements $A_{i,j}$ equal to either 1 or -1 
depending on the assumed direction of branch $i$ with respect to node $j$.

The incidence matrix allows writing the Kirchhoff Current Law as ${A}{I}=0$, where ${I}$ is a vector of branch currents. Furthermore, the vector of nodal voltages ${\phi}$ is related to the vector of branch voltages ${U}$ by $-{A}^{\top}{\phi}=U$.

In case of the LHC main dipole circuit during  nominal operation, i.e., when the power converter is active, only linear capacitors, inductors and resistors are considered, and the incidence matrix A can be decomposed into block form 
\begin{equation}
{A}=[A_{\text{C}}| A_{\text{R}} | A_{\text{L}} | A_{\text{CON}}]
\end{equation}
where $A_\text{R}$
is the resistance incidence matrix, $A_\text{C}$ 
the capacitance incidence matrix, $A_L$
the inductance incidence matrix, and $A_\text{CON}$ 
is the PID controller (voltage source) incidence matrix. Based on those matrices, a so-called conventional formulation of the Modified Nodal Analysis (MNA) 
\cite{Ho_1975aa} is expressed as
\begin{subequations}
\begin{align}
A_\text{C} C A_\text{C}^\top \partial_t \phi + A_\text{L} I_\text{L} + A_\text{R} G A_\text{R}^\top \phi + A_\text{CON} I_\text{con}&=0,\\
L \partial_t I_\text{L} - A_\text{L}^\top \phi &= 0\\
A_\text{CON}^\top\phi &= U_\text{con}
\end{align}
\label{eq:ModifiedNodalAnalysis}
\end{subequations}
with capacitance, inductance, and conductance matrices $C$, $L$, and $G$, respectively. The unknowns of the system are the vector of node potentials $\phi$, and currents through inductors 
$I_\text{L}$. 
In order to obtain a unique solution, 
an additional equation for the PI controller's unknown current in terms of the applied voltage drop is added and a grounding node is selected.

In analogy to \eqref{eq:ControllerInitialValueProblem}, the system \eqref{eq:ModifiedNodalAnalysis} is reformulated into an abstract differential-algebraic initial-value problem as
\begin{equation}
F_{\text{cir}}(\partial_t x_2, x_2, x_1)=0
\label{eq:CircuitInitialValueProblem}
\end{equation}
with initial value $x_2(t_0)=x_{2,0}$ and unknowns 
$x_2:=(\phi,I_{\text{L}},I_\text{con})^{\top}$. The external input function $x_1(t)$ is given in the studied case by the output voltage of the 
controller $U_\text{con}(t)$ which is obtained from $Y_{\text{meas}}(t)=I_\text{con}(t)$; see \eqref{eq:ControllerInitialValueProblem}.

Let us discuss the simple case of an inductance connected in series to a resistance and a controller (see Figure \ref{fig:simplecir}); the incidence matrices read 
$$
A_\text{C} \begin{bmatrix}0\\0\end{bmatrix},\;
A_\text{L} = \begin{bmatrix}0\\1\end{bmatrix},\; 
A_\text{R}=\begin{bmatrix}1\\-1\end{bmatrix},\;
A_\mathrm{CON}=\begin{bmatrix}-1\\0\end{bmatrix},
$$ 
with the elements $L=L_\mathrm{eq}$ and $G=R_\mathrm{eq}^{-1}$. The MNA yields four DAEs in the (redundant) variables $\phi_1$, $\phi_2$, $I_\text{con}$ and $I_\text{L}$. It is mathematically equivalently expressed by the ODE
\begin{equation}
	\label{eq:CircuitODE}
	\partial_t I_\text{con} = - \frac{R_\mathrm{eq}}{L_\mathrm{eq}}I_\text{con} - \frac{1}{L_\mathrm{eq}}U_\mathrm{con}.
\end{equation}
Recovering the underlying ODE is straightforward for small systems and if there are no LI cutsets or CV loops present in the circuit \cite{Estevez-Schwarz_2001aa}.

\begin{figure}
	\centering
	\includegraphics[width=0.3\textwidth]{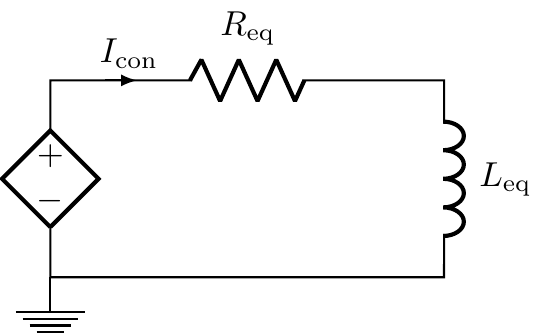}
	\caption{First order model of the LHC main dipole circuit.}
	\label{fig:simplecir}
\end{figure}
 
\section{Analysis of the Waveform Relaxation Scheme}
In practice, the circuit model \eqref{eq:CircuitInitialValueProblem} is usually numerically solved with an adaptive time stepping algorithm (e.g. trapezoidal rule or BDF schemes, \cite{Gunther_2000aa}) whereas the controller model employs fixed time steps corresponding to its sampling frequency. 
Coupling of both models via waveform relaxation allows to maintain this setup, by exchanging waveforms at communication time points $T_j$. For that purpose, the overall simulation time $\mathcal{I}$ is divided into $N$ time windows $\mathcal{I}_j = (T_j, T_{j+1}]$ with $j=0,...,N-1$. For each time window the 
two models are 
solved separately and their solutions (waveforms) are exchanged. The exchange of the waveforms can be organized such that (i) both models are provided with results from 
the previous iteration, (ii) one model is provided with a solution from the previous iteration and the other one with a solution calculated in the 
current iteration. The first method, called Jacobi scheme, is profitable for a parallel execution of models, whereas the second one, the Gauss-Seidel scheme, 
brings the benefit of faster convergence but it implies sequential execution of models \cite{Burrage_1995aa}. The termination criterion for the 
waveforms' exchange, i.e., for the 
convergence of 
the iterations in $\mathcal{I}_j$, is determined by applying an appropriate norm to measure the difference of two subsequent waveforms obtained with 
one or both of the models.

We employ the Gauss-Seidel scheme for the waveform relaxation method. At each time window, controller and circuit models are solved sequentially as depicted in Fig.~\ref{fig:WaveformRelaxationTechnique}. In order to incorporate the behavior of the delay $t_\mathrm{d}$ a zero order hold (ZOH) system is present in the main control loop \cite{Bartoszewski_2000aa}.

\begin{figure}[H]
\centering
\includegraphics[width=0.45\textwidth]{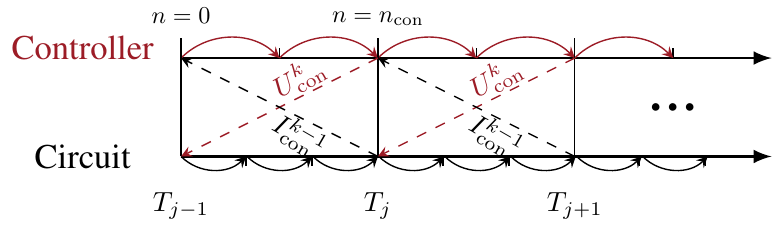}
\caption{\label{fig:WaveformRelaxationTechnique} Gauss-Seidel scheme for the information exchange between circuit (black) and controller (claret) models by means of the waveform relaxation technique. The controller is solved first, from the current calculated in the previous iteration. Then the obtained voltage is and input to the circuit model.}
\end{figure}

In the case of the considered coupled problem, with $j$ and $k$ denoting the time window and convergence iteration index, respectively, 
the waveform  relaxation algorithm takes the following steps:

\begin{enumerate}
\item[0)] Set $k=0$ and extrapolate $x_1^{(0)}(t)=x_{1,j,0}$ and $x_2^{(0)}(t)=x_{2,j,0}$ for $t\in\mathcal{I}_j$.

\item[1)] Solve (\ref{eq:ControllerInitialValueProblem}) with initial value $x_{1,j,0}$ and input $x_2^{(k)}(t)$ using fixed time steps to obtain 
$x_1^{(k+1)}(t)$ for $t\in\mathcal{I}_j$.

\item[2)] Solve (\ref{eq:CircuitInitialValueProblem}) with initial value $x_{2,j,0}$ and input $x_1^{(k+1)}(t)$ using an adaptive time stepping 
scheme 
to obtain $x_2^{(k+1)}(t)$ for $t\in\mathcal{I}_j$.

\item[3)] If not converged, then set $k=k+1$ and go to Step 1). Otherwise set $j=j+1$ and 
$$x_{1,j,0}=x_{1}^{(k)}(T_{j}) \text{  and  } x_{2,j,0}=x_{2}^{(k)}(T_{j}).$$ 
If $j<N$ proceed Step 0) to start the next time window.

\end{enumerate}
In the important case of ODEs, i.e., if there are no algebraic constraints 
\begin{align*}
	\mathrm{det}\frac{\partial F(\dot{y},y,u)}{\partial \dot{y}}\neq0,
\end{align*}
with $u(t)$, $y(t)$ being piecewise continuous functions and $F$ being globally Lipschitz continuous with respect to $y$ for all $u$, the Gauss-Seidel waveform relaxation algorithm 
converges uniformly on bounded time intervals (see \cite{Lelarasmee_1982ab,Burrage_1995aa}). This implies that the controller-circuit waveform relaxation scheme converges provided the circuit is given as the ODE \eqref{eq:CircuitODE}. This analysis can be extended to DAEs, e.g. \cite{Jackiewicz_1996aa,Bartel_2014ab} and even to delay differential equations \cite{Bartoszewski_2000aa}. However, the classical reasoning disregards errors due to time stepping since this error can be made arbitrarily small by reducing the step size $h$ in the respective models. This is not possible in the case of an external controller, where the time step is not given by numerical concerns but dictated by the actual hardware. The following analysis discusses this time-discrete setting. 

\begin{prop}
	Given the proposed waveform relaxation algorithm coupling equations \eqref{eq:odecontroller} and \eqref{eq:CircuitODE} with controller time step 
	$h_\mathrm{con} = t_\mathrm{d}$ and disregarding the error of the adaptive time stepping scheme used to solve \eqref{eq:CircuitODE}, then
	\begin{equation}
	I_{\text{con}}^{(k+1)}(t)=I_{\text{con}}^{(k)}(t), \; t\in\mathcal{I}_j
	\end{equation}
	for all $k\geq n_\mathrm{con}$.
\end{prop}
\begin{proof}
	Being $U_\mathrm{con}(t) = x_1(t)$ and $I_\mathrm{con}(t)=x_2(t)$ at $t\in\mathcal{I}_j$
	, a fine enough discretization of \eqref{eq:CircuitODE} 
	is assumed such  that the time 
	stepping error can be disregarded which yields 
	\begin{align}\label{eq:x2exact}
		x_2^{(k+1)}(t) =& e^{-\frac{R_\mathrm{eq}}{L_\mathrm{eq}}(t-t_{j,0})}\left(x_2(t_{j,0}) + 
		K_jx_1^{(k+1)}(t)\right),
	\end{align} 
	with operator $K_ju(t) =\int_{t_{j,0}}^te^{\frac{R_\mathrm{eq}}{L_\mathrm{eq}}(\tau-t_{j,0})}u(\tau)d\tau$.
	For equation \eqref{eq:odecontroller} an implicit Euler scheme is applied at time steps 
	\mbox{$T_j=t_{j,0}<...<t_{j,n_\mathrm{con}}=T_{j+1}$}. This, combined with equation \eqref{eq:x2exact} leads to
	\begin{align*}
		x_{1,j,n}^{(k+1)} =& x_{1,j,n-1}^{(k+1)} + h_\mathrm{con}k_p\partial_t{Y_\mathrm{ref}(t_{j,n-1})} \\
		&-\frac{h_\mathrm{con}k_\mathrm{p}}{L_\mathrm{eq}}x_{1,j,n-1}^{(k)}+ h_\mathrm{con}k_\mathrm{i}Y_\mathrm{ref}(t_{j,n-1}) \\
		&+\left(\frac{h_\mathrm{con}k_\mathrm{p}R_\mathrm{eq}}{L_\mathrm{eq}}-k_\mathrm{i}h_\mathrm{con}\right)
		e^{-\frac{R_\mathrm{eq}}{L_\mathrm{eq}}(t_{j,n-1}-t_{j,0})}\\
		&\left(x_2(t_{j,0})+ Kx_1^{(k)}(t)|_{t=t_{j,n-1}} \right),
	\end{align*}
	with $x_1^{(k+1)}(t_{j,n}) = x_{1,j,n}^{(k+1)}$.
	 It is sufficient 
	to show that for $k\geq n$,
	$x_{1,j,n}^{(k+1)}= x_{1,j,n}^{(k)}$, as then $x_2(t_{j,n})^{(k+1)} = x_2(t_{j,n})^{(k)}$, which implies the proposition. This follows via induction.
\end{proof}
An additional termination criterion can be introduced. The iteration stops when either (i) a difference between two 
subsequent waveforms is below a certain tolerance, or (ii) the convergence iteration index $k$ reached $n_\mathrm{con}$. Unless the current converged in 
less than 
$n_\mathrm{con}$ 
iterations, we can observe that the iteration-based stopping criterion requires one convergence iteration less with respect to the current-based 
convergence criterion. Additionally, if the time window size $H_j =T_{j+1}-T_j$ is equal to the controller sampling period, 
$h_{\text{con}}$, 
according to  Proposition 1, the algorithm reduces to a weak coupling scheme, i.e. both models are executed only once ($n_\mathrm{con}=1$) and the  
convergence check is not carried out.

For a certain class of open loop systems and reference profiles, the number of convergence iterations in the steady state is less than 
$n_\mathrm{con}$ as 
shown in Proposition 2.

\begin{prop} For systems of type \textit{p}, i.e. systems with \textit{p}-fold pole equal to zero in the continuous open loop 
system transfer function, with the reference signal being a polynomial function of time of (\textit{p}-1)-\textit{th} order, and for a given time 
window in the steady state, the proposed waveform relaxation algorithm converges at $k=1$.
\end{prop}

\begin{proof} The steady state error for systems of type \textit{p} with reference signals that are polynomials of time of order \textit{p}-1, is zero
\begin{equation}
e_{\text{ss}}(t)=s^{p-1}G_{\text{o}}(s)=s^{p-1}G_{\text{CON}}(s)G_{\text{CIR}}(s)=0,
\end{equation}
where $G_{\text{CON}}(s)$ and $G_{\text{CIR}}(s)$ are continuous transfer functions of explicit, input/output forms of 
(\ref{eq:ControllerInitialValueProblem}) and (\ref{eq:CircuitInitialValueProblem}), respectively. Since the steady state error is equal to zero, 
according to (\ref{eq:DiscretizedPIDController}), the controller output does not change 
$U_{\text{con}}^{(1)}(t_{j,n})=U_{\text{con}}^{(0)}(t_{j,n})$. 
The extrapolated controller output results in identical subsequent current profiles $I_{\text{con}}^{(1)}(t_{j,n})=I_{\text{con}}^{(0)}(t_{j,n})$, 
which concludes 
the proof.
\end{proof}

Proposition 2 indicates the minimum number of convergence iterations, and together with Proposition 1 provides upper and lower bounds for the number 
of convergence iterations during a transient. We obtain $2 \leq k \leq n_\mathrm{con}$.

The proposed waveform relaxation algorithm is now applied to a practical example of the co-simulation of LHC main dipole circuit and its power-converter controller. The next section deals with a selection of a PI controller gains based on an equivalent first-order approximation of the RB circuit. \section{PI Controller Design for a First Order Model of the LHC Main Dipole Circuit}
The goal of the current controller is to follow a predefined current profile $I_{\text{ref}}(t)$ corresponding to the energy increase of the 
particles in the accelerator. In other words, the increase of the magnetic field in the main dipole chain is synchronized with the increase of the 
energy of particles provided by the accelerating cavities. To ensure continuity of the first order derivative, i.e., 
$I_{\text{ref}} \in C^1(\mathcal{I})$, parabolic joints are introduced between linear sections whenever the current derivative is changing its value. 

For the controller design, we consider the LHC main dipole circuit, composed of 154 dipole magnets connected in series along with protection and 
powering devices. The main dipole circuit is represented as a first order system with equivalent inductance $L_{\text{eq}} = 154~L_{\text{dipole}} 
= 15.4 ~\text{H}$ and equivalent resistance $R_{\text{eq}} = 1~ \text{m}\Omega$. The design will be carried out in frequency domain by means of the 
Laplace transform for the closed loop system shown in Fig. \ref{fig:ClosedLoopSystem}. 

The transfer function of the first-order circuit model is given as
\begin{equation}
G_{\text{CIR}}(s) = \frac{I_{\text{con}}(s)}{U_{\text{con}}(s)} = \frac{1}{sL_{\text{eq}}+R_{\text{eq}}}
\label{eq:TransferFunctionRB}
\end{equation}
and the PI current controller transfer function reads
\begin{equation}
G_{\text{CON}}(s) = \frac{U_{\text{con}}(s)}{E(s)} = k_{\text{p}} + k_{\text{i}} \frac{1}{s}.
\label{eq:TransferFunctionPI}
\end{equation}
The open loop transfer function is
\begin{equation}
G_{\text{o}}(s) = G_{\text{CON}}(s)G_{\text{CIR}}(s) = 
\frac{k_{\text{p}} s + k_{\text{i}}}
{L_{\text{eq}}s^2+R_{\text{eq}}s}.
\label{eq:ContinuousTransferFunctionOpenLoop}
\end{equation}

The characteristic polynomial of the closed loop system has the following form
\begin{equation}
 s^2 + \frac{(R_{\text{eq}} +k_{\text{p}})}{L_{\text{eq}}}s + \frac{k_{\text{i}}}{L_{\text{eq}}} = 0.
\label{eq:CharacteristicPolynomialOpenLoop}
\end{equation}
The selection of the controller influences the roots of the polynomial and in consequence the dynamic response of the system. For that purpose we 
compare (\ref{eq:CharacteristicPolynomialOpenLoop}) with the normalized second order polynomial equation
\begin{equation}
 s^2 + 2\zeta\omega_0s + \omega_0^2 = 0,
\label{eq:NormalizedPolynomial2ndOrder}
\end{equation}
with $\zeta$ representing the relative damping, and where $\omega_0$ is the undamped natural frequency of the system. By comparing the coefficients of the respective polynomials (\ref{eq:CharacteristicPolynomialOpenLoop}) and (\ref{eq:NormalizedPolynomial2ndOrder}) we derive the relations for the controller gains
\begin{subequations}
\label{eq:controller_gains}
		\begin{align}
			 \omega_0^2-\frac{k_{\text{i}}}{L_{\text{eq}}}&=0,\label{eq:controller_gains_k_i}\\
			 \frac{(R_{\text{eq}} +k_{\text{p}})}{L_{\text{eq}}} - 2\zeta\omega_0 &= 0.\label{eq:controller_gains_k_p}
		\end{align}
\end{subequations}
For a practical application, the circuit should smoothly respond to a change in the reference current and oscillations are not desired. The typical 
value of relative damping is $\zeta=\frac{1}{\sqrt{2}} \approx 0.7071$. The value of natural frequency is determined by the bandwidth of the 
reference current profile and is equal to $f_{\text{bw}} = 1~ \text{Hz}$ leading to $\omega_0 = 2\pi f_{\text{bw}} = 6.28 ~\text{rad/s}$. 
\begin{figure}
\centering
\includegraphics[width=0.45\textwidth]{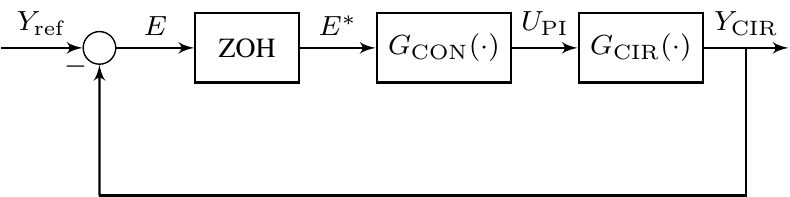}
\caption{\label{fig:ClosedLoopSystem}Closed loop feedback system with a discrete zero order hold (ZOH) element at the controller input}
\end{figure} 
\section{Numerical Examples}
The performance of the proposed waveform relaxation scheme is verified by means of two simulation scenarios compared to a monolithic 
reference simulation. For circuit and controller simulation we employ ORCAD Cadence PSpice and Powersim PSIM, respectively. The monolithic simulations are 
carried out with Powersim PSIM. According to the procedure described in Section IV (\ref{eq:controller_gains}), the controller gains are equal to 
$k_\text{p}=136.84$ and $k_\text{i}=607.97$. The controller sampling period is $t_\text{d}=0.04~\mathrm{s}$ is chosen as controller step size $h_{\text{con}}=t_\text{d}$ and the maximum time step size 
for the adaptive time-stepping algorithm of the circuit simulation is set to $h_{\text{cir}}\leq0.04~\mathrm{s}$ with absolute error tolerance equal to 
$10^{-10}$. 

\subsection{Step Response of a First-Order Model}
In the first test we verify the algorithm operation and validate Proposition 1 and 2. We consider a step response of the first-order model of the RB circuit co-simulated for $2.4$~s. The time interval $\mathcal{I}$ is divided into windows of fixed length $H=0.16~\mathrm{s}$ ($n_\text{con}=4$). 
Communication between the solvers occurs at discrete time instants $T_j=jH$. In order to demonstrate Proposition 1 and 2 we consider only the 
current-based 
termination 
criterion given as
\begin{equation*}
\frac{\int_{t_j}^{t_{j+1}}|I_{\text{cir}}^{(k)}(\tau)-I_{\text{cir}}^{(k-1)}(\tau)|\mathrm{d}\tau}
{\int_{t_j}^{t_{j+1}}|I_{\text{cir}}^{(k)}(\tau)|\mathrm{d}\tau}\leq10^{-6}, \text{for} \quad k\geq 1. 
\end{equation*}

Figure~\ref{fig:StepResponseCurrent} shows the comparison between the reference step current profile and the current responses of the co-simulated 
and monolithic simulations. The test shows a good agreement between results obtained in both cases.

Controller output voltages for the first time window are given in Table~1, corresponding to iteration steps, as well as to the monolithic simulation. As one can notice, at each consecutive iteration, the controller output $U_{\text{con,co-sim}}^k(t)$ approaches the monolithic solution 
$U_{\text{con,mono}}(t)$. In other words, each iteration $k$ extends the interval of time for which the controller output is equal to the monolithic reference 
solution, which is the typical behavior of the waveform relaxation scheme \cite{White_1985aa}. It is also worth noticing that the controller output for the last two 
iterations is identical, as the current-based convergence criterion requires two subsequent iterations with the same voltage inputs.

The number of iterations for each time window are reported in Fig.~\ref{fig:StepResponseControllerOutput}. During the initial transient, 
the number of convergence iterations is equal to $n_\text{con}+1$, whereupon it gradually decreases to two iterations, as the integral part of the PI controller calculates appropriate output for the studied circuit topology and excitation function. (Note that $k=0$ is the first iterate.) In this case, the open loop system is of type 1 as one can notice from (\ref{eq:ContinuousTransferFunctionOpenLoop}). To conclude, the test underlines the correctness of the number of iterations predicted by Propositions 1 and 2.

\begin{figure}
	\centering
	\includegraphics[width=0.5\textwidth]{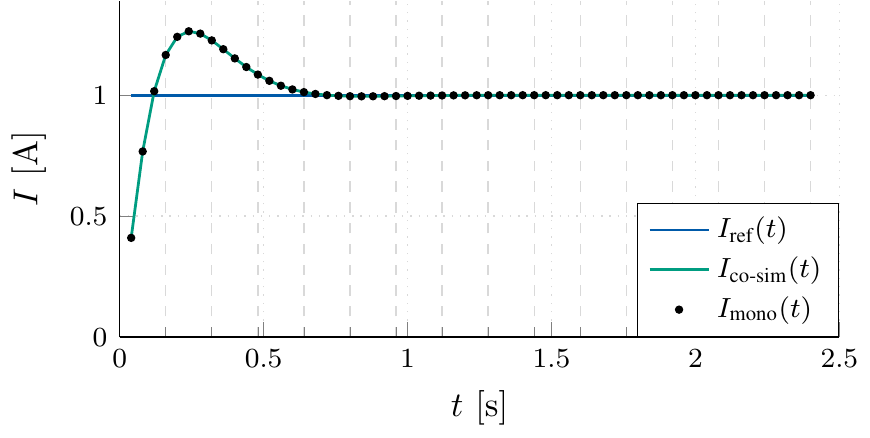}
	\vspace{-2em}
	\caption{Current evolution as a result of the step response of the PI controller and the first-order RB circuit model }
	\label{fig:StepResponseCurrent}
\end{figure}

\begin{table}
\caption{Controller output voltage for convergence iterations of the first time window of the co-simulation}
\vspace{-1em}
\fontsize{10}{10}\selectfont
\renewcommand{\arraystretch}{1.3}
\begin{center}
\begin{tabular}{ ccccc } 
 \hline 
 $t$ [s] & 0.04 & 0.08 & 0.12 & 0.16 \\
 \hline
 $U_\text{con, mono}(t)$ [V] & 161.16 & 119.34 & 76.12 & 41.67 \\ [0.5ex]
  \hline
 $U_\text{con, co-sim}^0(t)$ [V]  & 161.16 & 185.48 & 209.80 & 234.11 \\ [0.5ex]
 $U_\text{con, co-sim}^1(t)$ [V]  & 161.16 & 119.34 & 62.55 & -14.95 \\ [0.5ex]
 $U_\text{con, co-sim}^2(t)$ [V]  & 161.16 & 119.34 & 76.12 & 44.45\\ [0.5ex]
 $U_\text{con, co-sim}^3(t)$ [V]  & 161.16 & 119.34 & 76.12 & 41.67 \\ [0.5ex]
 $U_\text{con, co-sim}^4(t)$ [V]  & 161.16 & 119.34 & 76.12 & 41.67 \\ [0.5ex]
\hline
\end{tabular}
\end{center}
\label{table:proposition1}
\end{table}

\begin{figure}
	\centering
	\includegraphics[width=0.5\textwidth]{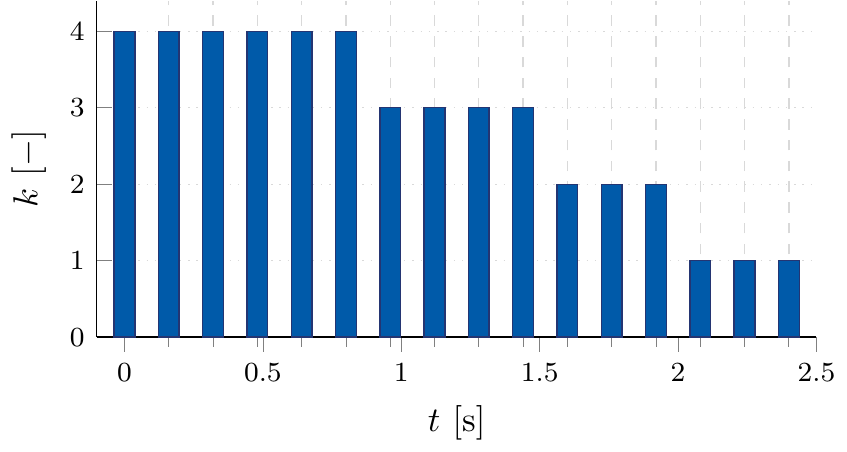}
	\vspace*{-2em}
	\caption{Number of convergence iterations per time window}
	\label{fig:StepResponseControllerOutput}
\end{figure}
 
\subsection{Parabolic-Linear Response of a Realistic Model}
After verifying the operation and convergence properties of the developed waveform relaxation coupling scheme, in the second test we perform a co-simulation of the considerably more complex LHC main dipole circuit composed of several thousand lumped components \cite{Ravaioli_2012aa}. The underlying ODE of this realistic model cannot be easily extracted and therefore this setup may not be covered by classical arguments of waveform relaxation analysis. 

The circuit is composed of 154 equivalent RLC models of a dipole magnet along 
with protection devices, a power converter and a filter. Parameters of the equivalent models were identified to match the frequency behavior of each 
of the magnets \cite{Ravaioli_2012ab}. We consider nominal operation of the circuit, i.e., quench protection systems are deactivated.
In this test, we verify the algorithm's operation as a weak coupling scheme. The co-simulation is performed for $120~$s and covers the initial parabolic part of the current profile, concatenated with the linear increase. The time interval $\mathcal{I}$ is divided into windows of fixed length $H=0.04~\mathrm{s}$ ($n_\text{con}=1$). 
The calculated current response accurately reproduces the reference current profile (see Fig.~\ref{fig:ParabolicLinearCurrentResponse}). The power converter controller output voltage is shown in Fig.~\ref{fig:ParabolicLinearVoltageResponse}.

\begin{figure}
	\centering
	\includegraphics[width=0.45\textwidth]{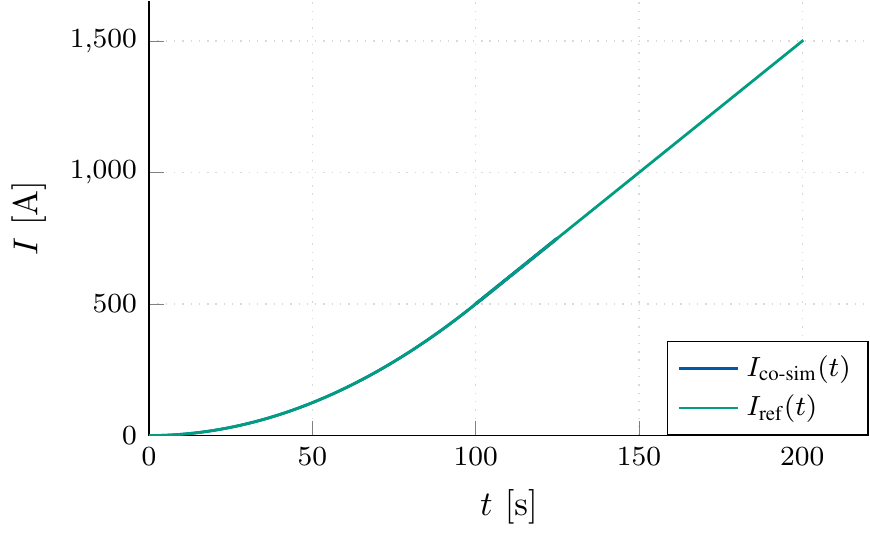}
	\vspace*{-1em}
	\caption{Initial part of the reference current profile composed of a parabolic increase ($0<t<100$) followed by  a linear ramp ($t>100$). Comparison of a parabolic-linear reference current and current solution for the controller/circuit coupling }
	\label{fig:ParabolicLinearCurrentResponse}
\end{figure}

\begin{figure}
	\centering
	\includegraphics[width=0.45\textwidth]{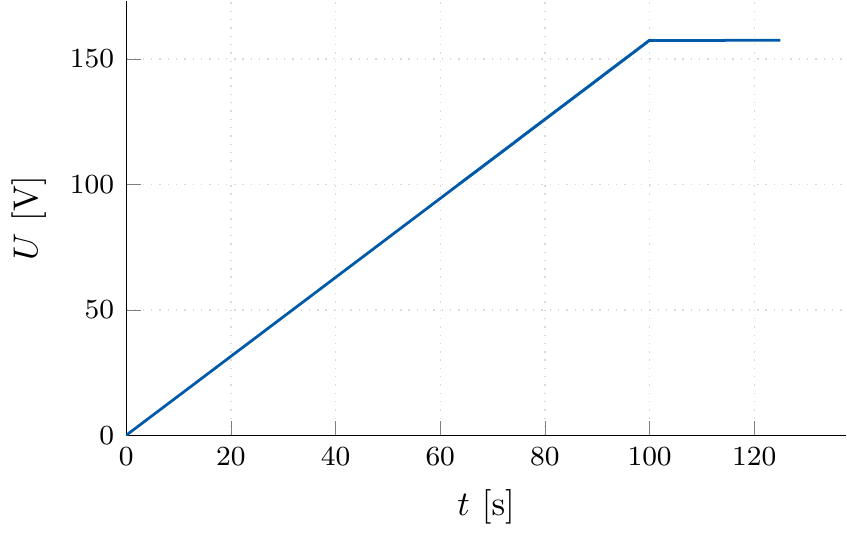}
	\vspace{-1em}
	\caption{PI controller output for the parabolic-linear reference current profile}
	\label{fig:ParabolicLinearVoltageResponse}
\end{figure}

\section{Conclusion}
In this paper, a waveform relaxation algorithm for the co-simulation of power-converter controller and electrical-circuit models has been proposed. 
The algorithm divides the total time interval into windows, and for each window both models are solved separately with an appropriate time-stepping 
algorithm. Properties of the waveform-relaxation algorithm's convergence have been studied and stated in two propositions. The first proposition defines the maximum number 
of convergence iterations leading to an additional termination criterion and potential reduction of the number of iterations with respect to the 
current-based stopping condition. The second proposition analyses the conditions for the minimum number of iterations. Boundedness of the proposed 
algorithm has been proven and demonstrated by means of numerical experiments. The algorithm has been also applied to a co-simulation of the 
power converter and the LHC main dipole circuit.

The proposed waveform relaxation algorithm will become an integral part of the STEAM framework and will allow for a more detailed analysis of  
protection circuits of superconducting magnets \cite{Bortot_2016ab}. One application consists of co-simulation of a power converter controller, a circuit, and the a quenching magnet. 
Such a model would allow to study the influence of the quench initiation on the power converter's response, and ensuing transients occurring in the circuit. A next 
step includes the analysis of the convergence of the electrical circuit with nonlinear elements and possibly partial differential equations.

\section*{Acknowledgment}

The authors would like to thank Samer Yammine from CERN for a discussion on the LHC main dipole power converter structure and characteristics.
This work has been partially supported by the Excellence Initiative of the German Federal and State Governments and the Graduate School of CE at TU Darmstadt.


\end{document}